\documentclass{article}
\usepackage[utf8]{inputenc}
\usepackage{amsthm, amsmath, amsfonts,mathtools, tikz, subcaption}
\usepackage[margin=1in]{geometry}
\usepackage[title]{appendix}
\usepackage[raggedrightboxes]{ragged2e}
\usepackage[english]{babel}
\usepackage[style=nature]{biblatex}
\usetikzlibrary{shapes.geometric, arrows, positioning}
\tikzstyle{condition} = [rectangle, text centered, draw=gray]
	\tikzstyle{from} = [<-, shorten <=1pt, >=stealth', semithick]
	\tikzstyle{into} = [->,shorten <=1pt, >=stealth', semithick]
	\tikzstyle{double} = [<->,shorten <=1pt, >=stealth', semithick]
	\tikzstyle{u} = [shape=circle, draw] 
	\tikzstyle{e} = [blue] 
	\tikzstyle{o} = [red]  
	\tikzstyle{i} = [xshift=5mm] 

\providecommand{\keywords}[1]
{
  \small	
  \textbf{\textit{Keywords: }} #1
}

\title{Transporting treatment effects from difference-in-differences studies}
\author{Audrey Renson$^1$, Ellicott C. Matthay$^1$, and Kara E. Rudolph$^2$ \\[2ex]
\small $^1$Department of Population Health, New York University Grossman School of Medicine, New York, United States\\%
\small $^2$Department of Epidemiology, Mailman School of Public Health, Columbia University,
New York, United States}
\date{
    \today
}

\newtheorem{assumption}{Assumption}
\newtheorem{theorem}{Theorem}

\newcommand\ind{\protect\mathpalette{\protect\independenT}{\perp}}
\def\independenT#1#2{\mathrel{\rlap{$#1#2$}\mkern2mu{#1#2}}}
\newcommand{\E}{\mathbb{E}}
\begin{document}

\maketitle

\begin{abstract}
    Difference-in-differences (DID) is a popular approach to identify the causal effects of treatments and policies in the presence of unmeasured confounding. DID identifies the sample average treatment effect in the treated (SATT). However, a goal of such research is often to inform decision-making in target populations outside the treated sample. Transportability methods have been developed to extend inferences from study samples to external target populations; these methods have primarily been developed and applied in settings where identification is based on conditional independence between the treatment and potential outcomes, such as in a randomized trial. We present a novel approach to identifying and estimating effects in a target population, based on DID conducted in a study sample that differs from the target population. We present a range of assumptions under which one may identify causal effects in the target population and employ causal diagrams to illustrate these assumptions. In most realistic settings, results depend critically on the assumption that any unmeasured confounders are not effect measure modifiers on the scale of the effect of interest (e.g., risk difference, odds ratio). We develop several estimators of transported effects, including g-computation, inverse odds weighting, and a doubly robust estimator based on the efficient influence function. Simulation results support theoretical properties of the proposed estimators. As an example, we apply our approach to study the effects of a 2018 US federal smoke-free public housing law on air quality in public housing across the US, using data from a DID study conducted in New York City alone.
\end{abstract}


\section{Introduction}
Difference-in-differences (DID) is a popular identification strategy when studying the causal effects of large-scale social and economic policies \cite{lechner2011estimation, roth2023s}. DID is appealing when: (i) randomization is not feasible, (ii) there is variation across jurisdictions and over time in terms of whether a policy was adopted, and (iii) not all variables that are confounders of the policy-outcome relationship are measured, leading to concerns about confounding bias 
\cite{imbens2009recent}. By comparing pre- and post-policy outcomes in both the jurisdiction implementing the policy and a comparable jurisdiction without the policy, and making a so-called parallel trends assumption (i.e., that changes in average potential outcomes over time are independent of policy adoption) \cite{lechner2011estimation,sofer2016negative}, DID can identify the causal effect of the policy on the outcome, even in settings where unmeasured variables would confound either (i) a pre-post analysis or (ii) a post-policy comparison between the treated and untreated jurisdictions.  

An important (often under-recognized) aspect of DID is that it identifies the average treatment effect among the treated (ATT) in the post-policy period, and not the average treatment effect (ATE) or other common parameters of interest \cite{lechner2011estimation}. For example, the ATT in a study of a policy raising the minimum wage is the effect of the policy on outcomes for the population living in the jurisdiction(s) that actually raised the minimum wage, and not the population living in all the jurisdictions in the study---those with and without the policy. ATT estimates resulting from a DID analysis are therefore most directly informative as to whether to maintain or discontinue policies in those locations. However, a major goal of studies that use DID is often to inform policy decisions by governments that have not yet adopted the policy of interest; in the minimum wage example, it may be of interest to inform decisions by the federal government or states with less generous minimum wage laws. Naïvely considering the estimated policy effects to apply to untreated jurisdictions requires the additional, strong assumption that there are no effect measure modifiers (measured or unmeasured) whose distribution varies between the treated jurisdiction(s) under study and the untreated jurisdiction(s) to which one wishes to make inferences \cite{lesko2017generalizing, pearl2014external}. For example, such an extrapolation would be biased if effects of the minimum wage differ by age, and age distributions differed across states. 

Generalizability and transportability methods have been developed with the goal of formally extending inferences made in one population to another population in the presence of effect heterogeneity \cite{cole2010generalizing,westreich2017transportability, rudolph2023efficiently}. These methods have mainly been applied in contexts where identification within the study sample is based on an unconfoundedness assumption, most typically in randomized controlled trials (RCTs). Inferences from RCTs generally apply only to the people participating in the trial, which may differ meaningfully from the true target population. We define ``target population'' to be the population to whom inference is desired, as dictated by substantive concerns. For example, in drug trials, patients at higher risk of the outcome are often oversampled to enhance power, and so effect estimated may be biased for the effect in the target population of patients who would receive treatments in practice \cite{greenhouse2008generalizing, stuart2015assessing}. Methods exist to quantitatively extend (i.e., transport or generalize) effects estimated in RCTs to target populations other than the included study sample, possibly alleviating the well-known tradeoff between internal and external validity in such studies \cite{colnet2020causal}.

It is plausible that transportablity methods could be used to quantitatively extend causal effects estimated from DID studies to target populations other than the treated sample, possibly alleviating the well-known tradeoff between internal and external validity in such DID studies as well. However, to our knowledge, neither identification assumptions nor estimators for transporting DID estimates have been addressed in the literature. DID presents special challenges for transportability because of the presumed existence of unmeasured confounders. Standard approaches to transportability assume that a conditional average treatment effect is constant between the sample and target population after conditioning on a measured set of covariates; if any unmeasured confounders in a DID application are also effect measure modifiers of the treatment-outcome relationship, then the existence of these unmeasured confounders creates complexities in evaluating this condition which have not, to our knowledge, been explored. Causal diagrams \cite{pearl1995causal} may facilitate such an exploration, as they have been essential in understanding assumptions for identification of transported effects \cite{pearl2011transportability, pearl2014external,rudolph2023efficiently}, but have seen limited use in DID settings \cite{sofer2016negative,caniglia2020difference}. This disconnect may be because causal diagrams generally only capture nonparametric independence assumptions \cite{pearl1995causal}, whereas parallel trends is a semiparametric assumption partially restricting the functional form of the outcome distribution \cite{roth2023parallel}.

This paper develops a formal approach to identification and estimation of effects in a target population, based on DID conducted in a study sample that differs from the target population. This paper is framed as transportability in the sense that we assume the study sample is not a subset of the target population \cite{dahabreh2020extending, westreich2017transportability}, though our results can easily be extended to the case where the study sample is nested within the target population. We employ causal diagrams to understand the sampling mechanism (i.e, the model that distinguishes the study sample from the target population), and show that our results rely crucially on the assumption that unmeasured confounders are either independent of the sampling process or are not effect modifiers on the scale of the effect being estimated (in this paper, we focus on additive effects such as the ATT and ATE, but our results can be generalized to non-additive measures, such as risk ratios). Section \ref{sec:prelim} describes the observed data and preliminary assumptions, Section \ref{sec:id} presents key identification results linking the observed data in the sample to causal quantities in the target population, and Section \ref{sec:est} presents estimators (including a doubly robust estimator based on the efficient influence function) for these quantities, which are illustrated using simulation in Section \ref{sec:sim}. Section \ref{sec:disc} concludes.

\section{Preliminaries}\label{sec:prelim}
Suppose we observe data on the variables $W_i, A_i, Y_{i0}$, and $Y_{i1}$ in a study sample containing $n$ individuals or units ($i=1,...,n$), where $W_i$ are (possibly multivariate) baseline covariates measured just before exposure, $A_i$ is a binary exposure, and $Y_{it}$ $(t=0,1)$ are outcomes measured before ($t=0$) and after ($t=1$) exposure occurs. Hereafter, we drop the $i$ subscript unless needed to resolve ambiguity. Suppose that the study sample is not representative of the true target population of interest, and that the latter contains $N$ individuals ($i=1,...,N$). We let $S=1$ denote membership in the study sample and $S=0$ denote membership in the target population. We assume that outcomes are only measured in the study sample, but that treatment and covariates are measured in both the study sample and the target population. Thus, the observed data take the form $O = \{S, A, W, Y_0S, Y_1S\}.$ Throughout, we use $f(x|\cdot)$ to denote a conditional density if $x$ is continuous and a conditional probability mass function if $x$ is discrete. Caligraphic uppercase letters denote the support of a random variable. 

We use $Y_t(a)$ to denote a potential outcome, or the outcome that would have occurred if exposure $A$ had been set by intervention to the value $a$. We assume the following throughout:
\begin{assumption}\label{asn:noint} (No interference)
    $Y_{it}(a_i, a_{i'})=Y_{it}(a_i)$ for $i\neq i'$, with $i,i'$ such that $\{S_i, S_{i'}\} \in \{0,1\}^2$
\end{assumption}
\begin{assumption}\label{asn:tvi} (Treatment version irrelevance)
    If $A_i=a$, then $Y_{it}=Y_{it}(a)$ with $i,i'$ such that $\{S_i, S_{i'}\} \in \{0,1\}^2$
\end{assumption}

Assumptions \ref{asn:noint} and \ref{asn:tvi} are standard in the causal inference and transportability literature and are not specific to the DID setting. Assumption \ref{asn:noint} requires that one unit's treatment does not impact another unit's potential outcome in either the sample or the target. Assumption \ref{asn:tvi} requires that treatments are sufficiently well-defined that observed outcomes can stand in for potential outcomes under treatment with the observed exposures, and that versions of the treatment do not differ between the sample and target. Assumptions \ref{asn:noint} and \ref{asn:tvi} are often referred to together as the stable unit treatment value assumption (SUTVA). Assumption \ref{asn:noint} can be relaxed in DID settings \cite{xu2023difference} but this is outside the scope of our work.

\subsection{Difference-in-differences in the study sample}
Here, we give a brief review of causal identification based on DID, which we will assume is the basis of identification in the study sample. Specifically, we invoke the following assumptions, standard in the DID literature \cite{heckman1997matching, abadie2005semiparametric, lechner2011estimation}: 
\begin{assumption}\label{asn:noant} (No anticipation): $Y_0(a)=Y_0$ for $a=0,1$
\end{assumption}
\begin{assumption}\label{asn:pos} (Positivity of treatment assignment)
    If $f(w|S=1)>0$ then $f(A=0|W=w, S=1)>0$ with probability 1 for all $w\in \mathcal{W}$
\end{assumption}
\begin{assumption}\label{asn:pt} (Parallel Trends): For $w\in \mathcal{W}$:
    	\[
    	\E\{Y_t(0)-Y_{t-1}(0)|A=1, S=1, W=w\}=\E\{Y_t(0)-Y_{t-1}(0)|A=0, S=1, W=w\} \]
\end{assumption}
Assumption \ref{asn:noant} states that future treatment does not impact the prior outcomes (this assumption can also be relaxed to allow anticipation up to a known time period \cite{callaway2021difference}). It is well known that under Assumptions \ref{asn:noint}-\ref{asn:pt}, it is possible to identify the $W$-conditional SATT, defined as $\eta(w)\equiv\E[Y_1(1) - Y_1(0)|W=w, A=1, S=1].$ Specifically, under Assumptions \ref{asn:noint}-\ref{asn:pt} we have:
\begin{align}\label{eq:did}
     \eta(w)
    &=\E[Y_1-Y_0|W=w, A=1,S=1]- \E[Y_1-Y_0|W=w, A=0, S=1] \\
    \nonumber &\equiv m_1(w) - m_0(w),
\end{align}
where we define $m_a(w) = \E[Y_1-Y_0|W=w, A=a, S=1]$. By extension, the unconditional sample ATT (abbreviated SATT, usually the focal parameter in DID) is identified as $\E[Y_1(1) - Y_1(0)|A=1,S=1]=\E[m_1(W) - m_0(W)|A=1,S=1].$ However, and importantly for our discussion, Assumptions \ref{asn:noint}-\ref{asn:pt} are not sufficient to identify parameters unconditional on $A=1$, such as the sample average treatment effect (SATE), defined as $\E[Y_1(1) - Y_1(0)|S=1].$ This is because parallel trends provides information about potential outcomes only among the treated group; without further assumptions there is no basis for identification of potential outcomes for the group $A=0$. Moreover, and as is the focus of this paper, additional assumptions would be required to identify effects outside the study sample, since parallel trends and positivity of treatment assignment are conditional on $S=1$.

\section{Identification of transported treatment effects}\label{sec:id}
In this section we consider the task of equating a causal estimand (i.e., one specified in terms of potential outcomes) in the target population to a function of the distribution of the observed data, $O$. Specifically, we focus on the population average treatment effect in the treated (PATT), defined as $\E[Y_1(1)-Y_1(0)|A=1, S=0]$, and the population average treatment effect (PATE) defined as $\E[Y_1(1)-Y_1(0)|S=0]$. We begin by introducing a motivating example, after which we introduce and discuss a set of sufficient identifying assumptions, and present identifying formulas which equal each causal estimand if the assumptions are true.
\subsection{Motivating example}
As of July 30, 2018, a US Department of Housing and Urban Development (HUD) rule required all public housing authorities to implement smoke-free housing (SFH) policies banning smoking in residences. We will consider the question: what effect did the federal SFH policy have on indoor air quality in US public housing developments? To answer this question, we consider transporting the results from a study conducted in public housing buildings in New York City (NYC) only. Specifically, a team of investigators conducted air quality monitoring in living rooms and common areas of NYC Housing Authority (NYCHA) public housing buildings, both before the federal policy went into effect (from April to July 2018), and again approximately every six months for 3 years post-policy. A DID analysis was conducted to estimate the effect of the policy on indoor air nicotine (among other measures), using as a comparison group a sample of households receiving housing assistance through a program known as Section 8, a public subsidy to supplement rental costs in private sector buildings \cite{anastasiou2023long}. Air quality was sampled in stairwells, hallways, and living rooms; for simplicity here we focus on stairwells. In addition to the air quality data, the team also conducted a baseline survey in a sample of residents, including information about tobacco smoking. (We note that in the original study, building inclusion criteria were high-rise [$>$15 floors], large resident population [$>$150 units], at least 80\% Black or Hispanic residents, and at least 20\% younger than 18 years; for simplicity we ignore these criteria here.) 

In this example, we let $i$ ($i=1,...n$) denote an individual participant in the baseline survey. Then $Y_{it}$ is a continuous variable representing the average of log-transformed air nicotine in stairwells in participant $i$'s building (where we let $t=0$ denote April-July 2018 and $t=1$ denote April-September 2021), $A_i=1$ denotes residence in a public housing building (and hence exposed to the SFH policy),  $A_i=0$ denotes residence in a Section 8 household (and hence unexposed),  and $W_i$ is a binary variable capturing whether participant $i$ reported any tobacco smoking inside the home in the prior 12 months. Thus, if Assumptions \ref{asn:noint}-\ref{asn:pt} hold (along with correct model specification and no measurement error), the DID results in this study may be interpreted as estimates of the effect of the SFH policy on indoor air quality for NYC public housing residents only. Though the study is informative as to the effect of the policy in NYC, it is also of interest to federal policymakers to estimate the PATT, which here represents the effect of the HUD rule on indoor air nicotine in April-September 2021 in public housing in the US. Data on the target population come from the American Housing Survey (AHS), which collects data on a biennial basis among a representative sample of US households; the 2015 sample includes information to construct $W_i$. We let $S_i=1$ denote participation in the baseline NYC sample survey, and $S_i=0$ denote participation in the AHS.

\subsection{Naïve approach}
We begin with an approach to transportability that does not take into account the causal structure of DID (in particular, does not take into account unmeasured confounding), after which we will use causal diagrams to illustrate why this approach will usually fail. Since all identified potential outcomes in equation (\ref{eq:did}) are conditional on $A=1$, an obvious starting point in attempting to transport effects identified through DID is to identify the PATT. Inspecting equation (\ref{eq:did}), a natural approach may be to assume that the $W$-conditional SATTs (conditional on each value of $w$) are equal between the sample and the target. If this were the case, one could identify the PATT using the following expression:
\begin{align}
    \nonumber \E[Y_1^1-Y_1^0|A=1,S=0] &= \E[ \E(Y_1^1-Y_1^0 |W, A=1, S=1)|A=1,S=0] \\
    &=\E[m_1(W)-m_0(W)|A=1,S=0] \label{eq:naive}
\end{align}
Specifically, in order for equation (\ref{eq:naive}) to hold, the following assumptions would be sufficient: 
\begin{assumption}\label{asn:exch_w} (Exchangeability of selection)
    $Y_t(a) \ind S | W, A=1$
\end{assumption}
\begin{assumption}\label{asn:pos_w} (Positivity of selection)
    If $f(w|A=1, S=0)>0$ then $f(S=1|W=w, A=1)>0$ with probability 1 for all $w \in \mathcal{W}$
\end{assumption}
Assumption \ref{asn:exch_w} states that, among the treated group, the distributions of potential outcomes in the sample and target are equal after conditioning on $W$. Assumption \ref{asn:pos_w} states that any covariate values that may occur in the target treated group must also be possible in the sample treated group.  Assumptions \ref{asn:exch_w} and \ref{asn:pos_w} together imply that the $W$-conditional ATT is constant across settings, and hence that the PATT is identified by equation (\ref{eq:naive}). 

Though Assumption \ref{asn:exch_w} is similar to the exchangeability of selection assumption usually invoked for transportability of the average treatment effect (ATE) (see e.g. Assumption 3 in \cite{colnet2020causal}), it differs importantly in that it must hold conditional on $A=1.$  This is so because DID was the basis for identification in the sample, so any effects identified in the sample (whether SATT or $W$-conditional SATT) are conditional on $A=1$, and the basis for transportability is therefore the constancy of the $W$-conditional SATTs (not SATEs) across settings. This constancy is dependent on replacing potential outcomes in the \textit{treated} target with those in the \textit{treated} sample, and for this replacement to be licensed, Assumption \ref{asn:exch_w} must condition on $A=1$.

Unfortunately, conditioning on $A=1$ means Assumption \ref{asn:exch_w} is unlikely to hold in most DID applications. To illustrate this point, Figure \ref{fig:swig1} displays a single world intervention graph (SWIG) \cite{richardson2013single} depicting a common DID setting. SWIGs are similar to causal directed acyclic graphs (DAGs) in that nodes represent random variables, directed arrows represent direct effects, and conditional independencies are given by \textit{d}-separation rules \cite{pearl1995causal}. SWIGs extend DAGs by depicting interventions on variables as split nodes ($A|a$ in Figure \ref{fig:swig1} indicates intervening to set $A=a$), and any variables affected by the intervention variable become potential outcomes under that intervention. Figure \ref{fig:swig1} represents a standard DID scenario in the sense that $U$ represents unmeasured common causes of $A$ and $Y_1$ that would confound a cross-sectional comparison, and whose existence motivates the use of DID. In Figure \ref{fig:swig2}, we add $S$ with arrows into $A$ and $W$, depicting the usual case that distributions of these variables differ across settings. (In our example, an approximately equal number of $A=0$ and $A=1$ buildings were selected purposively, altering the distribution of $A$.) Following the convention of selection diagrams, arrows emanating from $S$ represent ``exogenous conditions that determine the values of the variables to which they point'' \cite{pearl2011transportability}. Assumption \ref{asn:exch_w} would not be expected to hold in Figure \ref{fig:swig2} due to the existence of the path $S\rightarrow A \leftarrow U \rightarrow Y_1(a)$, on which $A$ is a collider. Thus, this path is opened by conditioning on $A$ (and not closed by conditioning on $W$), rendering $S$ potentially associated with $Y_1(a)$ conditional on $\{W, A=1\}$. For the same reason, $S$ is potentially associated with $Y_0$ conditional on $\{W, A=1\}$. Importantly, such paths will be present whenever (i) there is unmeasured confounding, and (ii) the target and sample differ in the distribution of treatment (conditional on $W$). We expect (i) to always be the case (otherwise DID would be unnecessary). We also expect (ii) to be the case except in rare circumstances such as when $A$ is experimentally assigned in both the target and the sample. Importantly, this failure of Assumption \ref{asn:exch_w} occurs regardless of whether the unmeasured confounders $U$ differ marginally in distribution between the target and sample (i.e., whether or not there is an arrow from $S$ into $U$ in Figure \ref{fig:swig2}). The issue of bias caused by conditioning on the treatment in transportability analyses has been discussed previously in the context of RCTs \cite{chiu2022selection}.

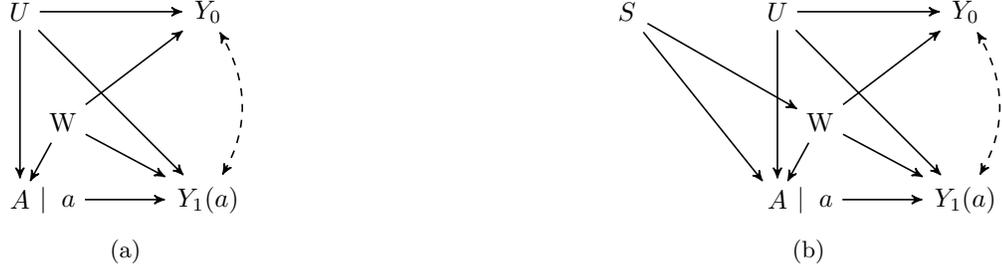
\begin{figure}
    \begin{subfigure}[t]{0.45\textwidth}
        \centering
        \begin{tikzpicture}[node distance=2.5cm]	
	\node (U) {$U$};
	\node (A) [below of=U] {$A$}
		edge[from] (U);
	\node (a) [i, below of=U] {$|\;\;a$};
 
        \node (W) [below right of=U, xshift=-1.2cm, yshift=3mm] {W}
            edge[into] (A);
 
	\node (Y0) [right of=U] {$Y_0$}
		edge[from] (W)
		edge[from] (U);
	\node (Y1) [below of=Y0] {$Y_1(a)$}
		edge[from] (a)
		edge[from] (W)
		edge[from] (U)
            edge[double, out=60, in=-60, dashed] (Y0);
    \end{tikzpicture}
        \caption{} \label{fig:swig1}
    \end{subfigure}
    \hfill
    \begin{subfigure}[t]{0.45\textwidth}
        \centering
        \begin{tikzpicture}[node distance=2.5cm]	
	\node (U) {$U$};
	\node (A) [below of=U] {$A$}
		edge[from] (U);
	\node (a) [i, below of=U] {$|\;\;a$};
 
        \node (W) [below right of=U, xshift=-1.2cm, yshift=3mm] {W}
            edge[into] (A);
 
	\node (Y0) [right of=U] {$Y_0$}
		edge[from] (W)
		edge[from] (U);
	\node (Y1) [below of=Y0] {$Y_1(a)$}
		edge[from] (a)
		edge[from] (W)
		edge[from] (U)
            edge[double, out=60, in=-60, dashed] (Y0);
        \node (S) [left of=U, xshift=5mm]{$S$}
            edge[into] (W)
            edge[into] (A);
    \end{tikzpicture}
        \caption{} \label{fig:swig2}
    \end{subfigure}
    \caption{Single world intervention graphs (SWIGs) illustrating common scenarios relevant for (a) DID in a sample and (b) transportability of DID results to a target population. Double-headed arrows indicate unmeasured common causes between two nodes. }\label{fig:swig}
\end{figure}

\subsection{Identification via restrictions on effect heterogeneity}
The analysis in the previous subsection illustrated that identification of transported effects based on exchangeability according to measured covariates (as in Assumption \ref{asn:exch_w}) is unlikely to be tenable in DID studies, since conditioning on $A=1$ will typically cause unmeasured covariates $U$ to be associated with the sampling mechanism $S$ regardless of whether this association exists marginally. Thus, Assumption \ref{asn:exch_w} would likely only be plausible if $U$ were included in the conditioning event, but this would not aid identification since $U$ is unmeasured. Fortunately, it is possible to identify transported effects when variables needed for exchangeability are unmeasured, so long as those variables are not also effect measure modifiers on the scale on which the causal effects are being measured, which we illustrate here. (As an aside, restrictions on effect heterogeneity have also been used for identification in instrumental variable approaches \cite{robins1994correcting}.)

We begin by expressing the concept that the existence of unmeasured confounders $U$ drive our decision to use DID by stating the following Assumptions, which relate only to identification in the sample:
\begin{assumption}\label{asn:exch_u_tx} (Latent exchangeability of treatment) $Y_t(a) \ind A|W, U, S=1$ for $t\in \{0,1\}$ and $a \in \{0, 1\}$
\end{assumption}
\begin{assumption}\label{asn:pos_u_tx} (Latent positivity of treatment)  If $f(u, w|S=1)>0$ then $f(A=a|U=u, W=w, S=1)>0$ with probability 1 for $a \in \{0,1\}$ and $\{u, w\} \in \{\mathcal{U}, \mathcal{W}\}$
\end{assumption}
Together, Assumptions \ref{asn:exch_u_tx} and \ref{asn:pos_u_tx} can be understood as stating that there is some set of unmeasured covariates $U$ such that, if measured, would (together with $W$) be sufficient to adjust for all confounding between the treated and untreated samples. In a sense, Assumption \ref{asn:exch_u_tx} does not introduce any new restrictions because one can define $U$ to be whatever variables (known or unknown) confound the cross-sectional association between $A$ and $Y_t$ and which motivate the use of DID in the first place. In contrast, Assumption \ref{asn:pos_u_tx} may be restrictive; the requirement that unmeasured confounding variables $U$ (known or unknown) have overlapping distribution between the treated and untreated may not hold in some settings and is not necessary for identification of the SATT via DID. (As an aside, it can be shown that parallel trends will hold if (i) Assumptions \ref{asn:exch_u_tx} and \ref{asn:pos_u_tx} hold and (ii) $U$ exerts a constant effect on $Y_0$ and $Y_1$ on the additive scale within levels of $W$ among the treated \cite{sofer2016negative}. However, in this paper we assume parallel trends to hold and do not consider what conditions render it plausible or not.) Next consider the follow assumptions aimed at identification in the target:
\begin{assumption}\label{asn:exch_u} (Latent exchangeability of selection) $Y_t(a) \ind S|W, U, A$ for $t\in \{0,1\}$ and $a \in \{0, 1\}$
\end{assumption}
\begin{assumption}\label{asn:pos_u} (Latent positivity of selection)
    If $f(u, w|A=a, S=0)>0$ then $f(S=1|U=u, W=w, A=a)>0$ with probability 1 for $a \in \{0,1\}$ and $\{u, w\} \in \{\mathcal{U}, \mathcal{W}\}$
\end{assumption}
Assumption \ref{asn:exch_u} modifies Assumption \ref{asn:exch_w} by allowing for $U$ in the conditioning event, so that the potential outcomes are equal in distribution between the sample and the target after conditioning on $W$, $U$, and $A$. Similarly, Assumption \ref{asn:pos_u} requires all possible values of both $U$ and $W$ in the target population to also be possible in the sample. In addition to conditioning on $U$, Assumptions \ref{asn:exch_u} and \ref{asn:pos_u} modify Assumption \ref{asn:exch_w} and \ref{asn:pos_w} by requiring their respective conditions for both the treated and untreated, not just the treated. This additional change allows us to identify more treatment effects including the PATE; if our focus were solely on the PATT, each assumption for $a=0$ only would be sufficient. Similarly to Assumption \ref{asn:exch_w}, we can assess Assumption \ref{asn:exch_u} graphically: if (as is the case in Figure \ref{fig:swig1}) the variables $\{W, U, A\}$ \textit{d}-separate $S$ from $Y_0$ and $S$ from $Y_1(a)$, then Assumption \ref{asn:exch_u} holds. 

Because $U$ is unmeasured, Assumptions \ref{asn:exch_u} and \ref{asn:pos_u} are insufficient for transportability; they render effects conditional on $\{W, U, A\}$ constant across settings, but these effects are not themselves identifiable. However, transportability is still possible if $U$ is not an additive effect measure modifier, which we state as follows:
\begin{assumption}\label{asn:homo} (U-homogeneity) $\E[Y_1(1)-Y_1(0)|U, W, S, A] =  \E[Y_1(1)-Y_1(0)|W, S, A]$
\end{assumption}
Note that Assumption \ref{asn:homo} does not require that $U$ not be a confounder, only that the treatment effect does not vary across levels of $U$ on the additive scale.  Note also the scale-dependence of Assumption \ref{asn:homo}; for example, it cannot hold for both log-transformed $Y_t$ and $Y_t$ on its natural scale, unless there is no effect of treatment or $U$ is unassociated with $Y_t$. The fact that Assumption \ref{asn:homo} refers to the additive scale follows from our focus is on additive treatment effects; if effects on an alternate scale (such as risk ratios) were of interest, then Assumption \ref{asn:homo} would need to be revised to express treatment effect homogeneity on that scale. 

If effects on the additive scale are homogeneous with respect to $U$, then additive effects conditional on $\{W, U, A\}$ (which are constant across settings by Assumptions \ref{asn:exch_u} and \ref{asn:pos_u}) do not depend on $U$, yielding identification of the PATT. This is stated in the following theorem:
\begin{theorem}\label{thm:trans_pt} (Transportability for difference-in-differences) Under Assumptions \ref{asn:noint}-\ref{asn:pt} and \ref{asn:exch_u}-\ref{asn:homo}, the PATT is identified as $\E[Y_1(1)-Y_1(0)|A=1,S=0] =\E[ m_1(W)-m_0(W) |A=1, S=0 ]$. Moreover, if Assumptions \ref{asn:exch_u_tx} and \ref{asn:pos_u_tx} also hold, then the PATE is identified as $\E[Y_1(1)-Y_1(0)|S=0] =\E[ m_1(W)-m_0(W) |S=0 ]$, and the population average treatment effect in the untreated (PATU) is identified as $\E[Y_1(1)-Y_1(0)|A=0, S=0] =\E[ m_1(W)-m_0(W) |A=0, S=0 ]$.
\end{theorem}
The proof of Theorem \ref{thm:trans_pt} is provided in Appendix \ref{app:trans_pt}. Importantly, Theorem \ref{thm:trans_pt} gives identifying formulas for the PATT as well as the PATE and PATU. Notably, Assumptions \ref{asn:exch_u_tx}-\ref{asn:pos_u_tx} are only required for identification of the PATE and PATU, not the PATT. This is an important distinction, particularly because one of the key advantages of DID is that identification can hold without having to assume positivity for the unmeasured confounders. As an aside, the addition of Assumptions \ref{asn:exch_u_tx} and \ref{asn:homo} to the standard identifying assumptions for DID (in our exposition, Assumptions \ref{asn:noint}-\ref{asn:pt}) also renders identifiable the SATE and the sample average treatment effect in the untreated (SATU) (shown in Appendix \ref{app:sate}). These results are intuitive: under latent exchangeability of selection, the treatment effects in the population are weighted averages of the $\{W,U,A\}-$conditional treatment effects in the sample; these conditional effects do not depend on $U$ under $U$-homogeneity. Moreover, because $U$ represents all unmeasured confounders, differences between the $W$-conditional SATT, SATE, and SATU can only be caused by effect heterogeneity according to $U$, which has been ruled out by Assumption \ref{asn:homo}. Therefore the $W$-conditional SATT, SATE, and SATU all equal one another.

Assumptions \ref{asn:exch_u_tx}-\ref{asn:homo} are not the only set of assumptions that yield identification of effects in the target when $U$ is related to the sampling mechanism, but alternative assumption sets will generally also place restrictions on unmeasured effect heterogeneity. For example, supposing that Assumptions \ref{asn:exch_u}-\ref{asn:pos_u} hold, it is possible to identify the PATT under a parallel trends assumption for both the treated and untreated counterfactual regimes (i.e., if we added to Assumption \ref{asn:pt} an equivalent expression replacing $Y_t(0)$ and $Y_{t-1}(0)$ with $Y_t(1)$ and $Y_{t-1}(1)$). However, this stronger parallel trends assumption also implies the $W$-conditional SATT, SATE, and SATU are all equal \cite{callaway2021difference}, implying effect homogeneity according to $U$. 

Table \ref{tab:assumptions} provides interpretations of each of the 12 Assumptions presented in terms of the applied question. In particular, $U$ represents unmeasured differences between public housing and Section 8 housing that impact levels of air nicotine independently of the treatment, leading investigators to pursue a DID design. For example, $U$ may represent ventilation (with $U=1$ denoting high and $U=0$ denoting low ventilation); we expect public housing buildings to more often have low ventilation and that ventilation impacts air nicotine, but ventilation was not measured in the study. In Figure \ref{fig:swig2}, arrows from $S$ into $W$ and $A$ depict measured environmental and societal conditions that lead to differing air quality and differing distributions of public housing vs. Section 8 residence across regions in the US. Since $Y_t$ was log-transformed, Assumption \ref{asn:homo} requires that for residents with the same levels of indoor smoking and separately for public housing and Section 8, the additive effect of a smoke-free housing policy on log-transformed air nicotine (and hence a type of multiplicative effect) is constant for buildings with high and low ventilation. Thus, Assumption \ref{asn:homo} would be violated if high- and low-ventilation buildings had differing baseline levels of air nicotine and the effect of a smoke-free housing policy was to decrease air nicotine by a constant absolute amount (e.g., a constant reduction in parts per million).

\begin{table}
    \centering
    \begin{tabular}{{|p{1.5in}|p{4in}|}}
    \hline
       Assumption  & Meaning in application \\ 
       \hline
       \ref{asn:noint} (No interference) & Air nicotine levels for one individual are not impacted by SFH policy exposure of other individuals \\
       &\\
       \ref{asn:tvi} (Treatment version irrelevance) & Variation in SFH implementation/enforcement do not affect air nicotine \\
       &\\
       \ref{asn:noant} (No anticipation) & Individuals did not change their behavior in anticipation of SFH \\
       &\\
       \ref{asn:pos} (Positivity of treatment assignment) & Among individuals both with and without indoor smoking in the last 12 months, there exist Section 8 residents \\
       &\\
       \ref{asn:pt} (Parallel trends) & In absence of the SFH policy, at a given level indoor smoking in the last 12 months, absolute changes in log-transformed air nicotine levels over time would have been equal between NYCHA and Section 8 residents  \\
       &\\
       \ref{asn:exch_w} (Exchangeability of selection) & Among public housing buildings nationally at a given level indoor smoking in the last 12 months, the distribution of potential nicotine levels is equal between NYC and the rest of the USA\\
       &\\
       \ref{asn:pos_w} (Positivity of selection) & There exist  NYCHA residents who both did and did not report indoor smoking \\
       &\\
       \ref{asn:exch_u_tx} (Latent exchangeability of treatment) & In NYC only, at a given level of indoor smoking and ventilation, potential air nicotine distributions are independent of Section 8 vs. NYCHA residency\\
       &\\
       \ref{asn:pos_u_tx} (Latent positivity of treatment) & There are both NYCHA and Section 8 buildings at all levels of indoor smoking and ventilation present in the study sample\\
       &\\
       \ref{asn:exch_u} (Latent exchangeability of selection) & Among public housing buildings nationally, at a given level of indoor smoking and ventilation, the distribution of potential nicotine levels is equal between NYC and the rest of the USA\\
       &\\
       \ref{asn:pos_u} (Latent positivity of selection) & There are NYCHA buildings at all levels of indoor smoking and ventilation present for public housing buildings nationally\\
       
       &\\
       \ref{asn:homo} (U-homogeneity) & Building ventilation is not an additive effect measure modifier after conditioning indoor smoking and indicators of NYC vs. remaining USA and Section 8 vs. public housing. \\
       \hline
    \end{tabular}
    \caption{Interpretation of assumptions in the motivating example}
    \label{tab:assumptions}
\end{table}

\section{Estimators}\label{sec:est}
In the section, we presume identification holds according to one of the sets of assumptions presented Theorem \ref{thm:trans_pt}, and consider the problem of estimating the statistical parameter 
\[\psi(a^*)=\E[ \eta(W) |A=a^*, S=0 ], a^*=0,1.\] (See equation (\ref{eq:did}) for the definition of $\eta(\cdot)$.) From Theorem \ref{thm:trans_pt}, we have that under Assumptions \ref{asn:noint}-\ref{asn:pt} and \ref{asn:exch_u}-\ref{asn:homo}, $\psi(1)$ equals the PATT; with the addition of Assumptions \ref{asn:exch_u_tx}-\ref{asn:pos_u_tx}, $\psi(0)$ equals the PATU and $E[\psi(A)|S=0]$ equals the PATE. 
To simplify notation in this section, let $\Delta Y=Y_1-Y_0$ denote differenced outcomes, $m_a(W)=\E[\Delta Y|W,S=1, A=a]$  denote the true outcome-difference model, and $g_{a,s}(W)=f(A=a, S=s|W)$, denote the true propensity scores for treatment assignment and selection. We use $\widehat m_a$ and $\widehat g_{a,s}$ to denote estimators of those quantities, which may or may not be correctly specified. A correctly specified model is one that converges in probability to the true population moments. We also use $P_n \{ h(O) \}=n^{-1} \sum_{i=1}^n h(O_i)$ to denote the sample average of a function $h(\cdot)$ of the observed data.
\subsection{G-computation estimator}
A g-computation estimator (also called a substitution estimator or plug-in estimator) is constructed by plugging in estimators of the empirical counterparts of the population quantities into the identifying formula in Theorem \ref{thm:trans_pt}:
\begin{align*}
    \widehat\psi_{gcomp}(a^*) &= 
    P_n \bigg \{ 
    \frac{I(A=a^*, S=0)}{P_n \{ I(A=a^*, S=0)\} }\{ \widehat m_1(W) - \widehat m_0(W)\}\bigg\} 
\end{align*}
The estimator $\widehat\psi_{gcomp}$ will be consistent and asymptotically normal if $\widehat m_a$ is correctly parametrically specified, but not necessarily otherwise.
\subsection{Inverse-odds weighted estimator}
Instead, one may have more information about the functional form of the propensity scores. The following inverse-odds weighted estimator will be consistent and asymptotically normal if $\widehat g_{a,s}$ are correctly parametrically specified, but not necessarily otherwise:
\begin{align*}
       \widehat\psi_{iow}(a^*) &= P_n \bigg\{ \bigg[ \frac{I(A=a^*, S=1) \widehat g_{a^*,0}(W)}{P_n\{ I(A=a^*, S=0)\}\widehat g_{1,1}(W)} - \frac{I(A=0, S=1) \widehat g_{a^*,0}(W)}{ P_n\{ I(A=a^*, S=0)\} \widehat g_{0,1}(W)} \bigg] \Delta Y \bigg\}
\end{align*}

\subsection{Doubly robust estimator} 
Lastly, we provide a doubly robust estimator, meaning in this case that the estimator is consistent and asymptotically normal if either $\widehat g_{a,s}$ or $\widehat m_a$ consistent of correctly-specified parametric models; it need not be the case that both are correct. A doubly robust estimator for $\psi$ is given by: 
\begin{align*}
   \widehat\psi_{dr}(a^*) &= P_n \bigg\{ \frac{I(A=a^*, S=1) \widehat g_{a^*,0}(W)}{ P_n I(A=a^*, S=0)\widehat g_{1,1}(W)}
        \{\Delta Y - \widehat m_1(W) \} \\
        &\qquad - \frac{I(A=0, S=1) \widehat g_{1,0}(W)}{P_nI(A=a^*, S=0) \widehat g_{0,1}(W_i)}
        \{\Delta Y - \widehat m_0(W) \} \\
    &\qquad + \frac{I(A=a^*,S=0)}{P_nI(A=a^*, S=0)}\{\widehat m_1(W) - \widehat m_0(W)  \} \bigg\}
\end{align*}
In Appendix \ref{app:eif}, we show the derivation of $\widehat\psi_{dr}(a^*)$ as a ``one-step'' estimator based on the efficient influence function for $\psi(a^*)$, which implies that $\widehat\psi_{dr}(a^*)$ is asymptotically efficient. The fact that $\widehat\psi_{dr}(a^*)$ corresponds to the efficient influence function also leads to an estimator of the asymptotic variance under the assumption that $\widehat g_{a,s}$ and $\widehat m_a$ are both correctly specified, which we also provide in Appendix \ref{app:eif}. In Appendix \ref{app:dr}, the double robust property of $\widehat\psi_{dr}(a^*)$ is demonstrated, and proof of the consistency of the g-computation and IOW estimators are provided as a bi-product of the double robust property. Code to implement the proposed estimators is available at \verb|https://github.com/audreyrenson/did_generalizability|.
\section{Simulation study}\label{sec:sim}
We generated $nsims=200$ datasets of $nobs=10,000$ each, according to the following data generating mechanism:
\begin{align*}
    S &\sim Bernoulli(0.5) \\
    U &\sim Bernoulli( logit^{-1}[-1 + S]) \\
    W &\sim Bernoulli(0.5 - 0.25S) \\
    A &\sim Bernoulli(0.3 + 0.1S + 0.1W + 0.1U) \\
    Y_0 &\sim N(1 + W + U, 0.1) \\
    Y_1 &\sim N(0.5W + U + A + 0.5WA, 0.1)
\end{align*}
To see that parallel trends holds in the simulation, note that for $a=0,1$: 
\begin{align*}
    \E[&Y_1(0)-Y_0(0)|A=a, S=1, W] \\
    &= \E\{[0.5W + U + (0) + 0.5W(0)] - [1+W+U]|A=a, S=1, W\} \\
    &= \E\{[0.5W  - 1 - W]|A=a, S=1, W\} \\
    &= -1 - 0.5W
\end{align*}
We applied each of the three proposed estimators for the PATT to each dataset with all models correctly specified, all models incorrectly specified, only outcomes models misspecified, and both selection and treatment models misspecified. We treated $U$ as an unmeasured variable in all analysis. For correctly specified models, all variables except $U$ were included with the above functional form, in misspecified outcome models we only include main terms for W and A, and in misspecified propensity models we dropped terms for S. The true PATT$=1.28$ was calculated by generating potential outcomes for 1 million observations. Results shown in Figure \ref{fig:sim_results} illustrate that IOW is biased whenever the propensity score for treatment and selection is misspecified, g-compuation is biased whenever the outcome model is misspecified, and that the doubly robust estimator is approximately unbiased  if either model is correct. Code to implement the simulation is available at \verb|https://github.com/audreyrenson/did_generalizability|.
\begin{figure}
    \centering
    \includegraphics[scale=1]{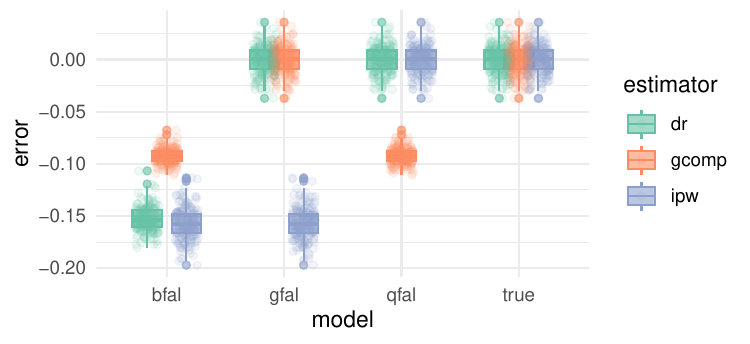}
    \caption{Results of each estimator applied to $nsims=200$ simulated datasets of $nobs=10,000$ each. $bfal$ indicates all models misspecified,  $qfal$ only outcomes models misspecified, $gfal$ both selection and treatment models misspecified, and $true$ all models correctly specified.}
    \label{fig:sim_results}
\end{figure}


\section{Example analysis}
Data on outcomes $Y_{it}$ in the study sample are measured from various monitors throughout the building and are not associated with a particular baseline survey respondent. Thus, we use building-level averages to represent outcomes for each respondent.  We concatenated data from the NYC SFH baseline survey with data from the 2015 AHS and coded a variable $S$ to track which dataset an observation came from. To account for the complex survey design of the AHS, we used an approach similar to those applied in previous generalizability/transportability analyses using survey weights \cite{ackerman2021generalizing, rudolph2022efficiently}. Specifically, respondents in the study sample were assigned survey weights of 1, respondents in the AHS sample were assigned their supplied design weights, and these weights were applied to all models and estimators. All models were saturated since $W$ is binary. We then constructed IOW, g-computation, and DR estimators based on these models, where in all cases the sample mean $P_n$ is taken with respect to the weighted, combined NYC SFH and AHS samples. Standard errors for these estimators were constructed using the jackknife replicate weights provided with the AHS data and a nonparametric bootstrap within the NYC SFH sample (we used 160 bootstrap replications to match the number of jackknife resamples). Code and data to implement the example analysis are provided at \verb|https://github.com/audreyrenson/did_generalizability|.

The stacked dataset contained $n=13,156$ observations, of which $n=1,029$ were from the NYC SFH study ($S=1$) and $n=12,127$ were from the AHS ($S=0$). The estimated SATT (95\% confidence interval [CI]) was -1.50 (-1.70, -1.31) (estimated from a simple unadjusted DID model). The estimated PATT was  -1.59 (-1.58, -1.60) using g-computation, -18.9 (-20.4, -17.3) using IOW, and -1.59 (-2.43, -0.75) using the DR estimator. That the transported estimates are more extremely negative than the SATT estimate reflects the facts that (i) the stratified effect estimate among those reporting indoor smoking (-1.97) is more extreme than those not reporting indoor smoking (-1.47), and the estimated proportion reporting indoor smoking is greater in the target population than the study sample (24\% vs. 7\%). 

\section{Discussion}\label{sec:disc}
This paper introduced an approach to estimating treatment effects in a target population based on DID conducted in a study sample that differs from the target population. Under certain assumptions, some of which may be understood with the aid of causal diagrams, we can identify the PATT, PATE, and PATU. We also propose several estimators of the aforementioned effects in the target population that only require measurement of covariates and/or treatments in the target population, not necessarily outcomes. This approach may be useful when, as is the case in our motivating example involving air nicotine, measurement of outcomes in the target population (in this case, the entire U.S.) may not be feasible, and unobserved confounding is present (in this case, ventilation) but those unobserved confounders do not modify the additive treatment effect. Though our approach assumed the same set of covariates were sufficient for internal and external validity, the methods can easily be adapted to settings where the covariates needed for external validity are a subset of those needed for internal validity. Though our approach has been framed around the problem of transportability (i.e., the study sample is not a subset of target population), our methods can easily be adapted to generalizability problems (when the study sample is nested in the target population). The approach may therefore also prove useful when a select group of jurisdictions (such as states or provinces) implement a policy, but decisions need to made a higher level of organization (such as national governments). 

It is important to note that our motivating example was greatly simplified for illustrative purposes; a full analysis to address the motivating question would likely be more complex. For example, one would need to carefully consider how the exclusion criteria may impact the plausibility of assumptions, whether other covariates would need to be measured, and whether differences in building management between public housing in NYC and other areas might violate treatment version irrelevance.

Causal diagrams have rarely been employed to understand identification in DID designs, but have been essential for elucidating the way causal structure impacts generalizability and transportability problems. By employing causal diagrams, we highlighted that the causal structure implied by unmeasured confounding that often motivates DID creates particular complexities for generalizability and transportability. Specifically, we were able to identify transported treatment effects under an assumption that the unmeasured confounders are not additive effect measure modifiers, but not necessarily otherwise.

The validity of an assumption that unmeasured confounders are not additive effect measure modifiers may be difficult to assess in practice. In our example, we possess no \textit{a priori} substantive information to suggest that the additive effect of a smoking ban on log-transformed air nicotine would be constant according to the building's level of ventilation (a presumed confounder). This suggests that, when transportability is of interest in DID studies, investigators should measure and adjust for as many potential confounders as possible (even if not formally needed for parallel trends) in order to reduce the number of variables for which we must make homogeneity assumptions. Future work will seek to develop bounds under violations of effect homogeneity along with methods to assess the sensitivity of conclusions to this key assumption. Additionally, derivations of asymptotic properties of the proposed estimators relied on parametric models for nuisance functions, which may be infeasible in practice. Doubly robust estimators of other estimands can often be shown to retain asymptotic normality when slower-converging, data-adaptive nuissance function estimators are used, though the same is not case for singly robust estimators based on g-computation or inverse weighting \cite{van2014higher,rotnitzky2021characterization}. Thus, the extension of our methods to relax model specification assumptions will also be of interest in future work.

\printbibliography

@article{heckman1997matching,
  title={Matching as an econometric evaluation estimator: Evidence from evaluating a job training programme},
  author={Heckman, James J and Ichimura, Hidehiko and Todd, Petra E},
  journal={The review of economic studies},
  volume={64},
  number={4},
  pages={605--654},
  year={1997},
  publisher={Wiley-Blackwell}
}

@article{abadie2005semiparametric,
  title={Semiparametric difference-in-differences estimators},
  author={Abadie, Alberto},
  journal={The review of economic studies},
  volume={72},
  number={1},
  pages={1--19},
  year={2005},
  publisher={Wiley-Blackwell}
}

@article{anastasiou2023long,
  title={Long-term trends in secondhand smoke exposure in high-rise housing serving low-income residents in New York City: three-year evaluation of a federal smoking ban in public housing, 2018--2021},
  author={Anastasiou, Elle and Gordon, Terry and Wyka, Katarzyna and Tovar, Albert and Gill, Emily and Rule, Ana M and Elbel, Brian and Kaplan, JD Sue and Shelley, Donna and Thorpe, Lorna E},
  journal={Nicotine and Tobacco Research},
  volume={25},
  number={1},
  pages={164--169},
  year={2023},
  publisher={Oxford University Press US}
}

@article{callaway2021difference,
  title={Difference-in-differences with a continuous treatment},
  author={Callaway, Brantly and Goodman-Bacon, Andrew and Sant'Anna, Pedro HC},
  journal={arXiv preprint arXiv:2107.02637},
  year={2021}
}

@article{caniglia2020difference,
  title={Difference-in-difference in the time of cholera: a gentle introduction for epidemiologists},
  author={Caniglia, Ellen C and Murray, Eleanor J},
  journal={Current epidemiology reports},
  volume={7},
  pages={203--211},
  year={2020},
  publisher={Springer}
}

@article{cole2010generalizing,
  title={Generalizing evidence from randomized clinical trials to target populations: the ACTG 320 trial},
  author={Cole, Stephen R and Stuart, Elizabeth A},
  journal={American journal of epidemiology},
  volume={172},
  number={1},
  pages={107--115},
  year={2010},
  publisher={Oxford University Press}
}

@article{colnet2020causal,
  title={Causal inference methods for combining randomized trials and observational studies: a review},
  author={Colnet, B{\'e}n{\'e}dicte and Mayer, Imke and Chen, Guanhua and Dieng, Awa and Li, Ruohong and Varoquaux, Ga{\"e}l and Vert, Jean-Philippe and Josse, Julie and Yang, Shu},
  journal={arXiv preprint arXiv:2011.08047},
  year={2020}
}

@article{dahabreh2020extending,
  title={Extending inferences from a randomized trial to a new target population},
  author={Dahabreh, Issa J and Robertson, Sarah E and Steingrimsson, Jon A and Stuart, Elizabeth A and Hernan, Miguel A},
  journal={Statistics in medicine},
  volume={39},
  number={14},
  pages={1999--2014},
  year={2020},
  publisher={Wiley Online Library}
}

@article{greenhouse2008generalizing,
  title={Generalizing from clinical trial data: a case study. The risk of suicidality among pediatric antidepressant users},
  author={Greenhouse, Joel B and Kaizar, Eloise E and Kelleher, Kelly and Seltman, Howard and Gardner, William},
  journal={Statistics in medicine},
  volume={27},
  number={11},
  pages={1801--1813},
  year={2008},
  publisher={Wiley Online Library}
}

@article{hines2022demystifying,
  title={Demystifying statistical learning based on efficient influence functions},
  author={Hines, Oliver and Dukes, Oliver and Diaz-Ordaz, Karla and Vansteelandt, Stijn},
  journal={The American Statistician},
  volume={76},
  number={3},
  pages={292--304},
  year={2022},
  publisher={Taylor \& Francis}
}

@article{lechner2011estimation,
  title={The estimation of causal effects by difference-in-difference methods},
  author={Lechner, Michael and others},
  journal={Foundations and Trends{\textregistered} in Econometrics},
  volume={4},
  number={3},
  pages={165--224},
  year={2011},
  publisher={Now Publishers, Inc.}
}

@article{lesko2017generalizing,
  title={Generalizing study results: a potential outcomes perspective},
  author={Lesko, Catherine R and Buchanan, Ashley L and Westreich, Daniel and Edwards, Jessie K and Hudgens, Michael G and Cole, Stephen R},
  journal={Epidemiology (Cambridge, Mass.)},
  volume={28},
  number={4},
  pages={553},
  year={2017},
  publisher={NIH Public Access}
}

@article{imbens2009recent,
  title={Recent developments in the econometrics of program evaluation},
  author={Imbens, Guido W and Wooldridge, Jeffrey M},
  journal={Journal of economic literature},
  volume={47},
  number={1},
  pages={5--86},
  year={2009},
  publisher={American Economic Association}
}

@article{pearl1995causal,
  title={Causal diagrams for empirical research},
  author={Pearl, Judea},
  journal={Biometrika},
  volume={82},
  number={4},
  pages={669--688},
  year={1995},
  publisher={Oxford University Press}
}

@inproceedings{pearl2011transportability,
  title={Transportability of causal and statistical relations: A formal approach},
  author={Pearl, Judea and Bareinboim, Elias},
  booktitle={Proceedings of the AAAI Conference on Artificial Intelligence},
  volume={25},
  number={1},
  pages={247--254},
  year={2011}
}

@article{pearl2014external,
  title={External Validity: From Do-Calculus to Transportability Across Populations},
  author={Pearl, Judea and Bareinboim, Elias},
  journal={Statistical Science},
  volume={29},
  number={4},
  pages={579--595},
  year={2014}
}

@article{richardson2013single,
  title={Single world intervention graphs (SWIGs): A unification of the counterfactual and graphical approaches to causality},
  author={Richardson, Thomas S and Robins, James M},
  journal={Center for the Statistics and the Social Sciences, University of Washington Series. Working Paper},
  volume={128},
  number={30},
  pages={2013},
  year={2013},
  publisher={Citeseer}
}

@article{roth2023s,
  title={What’s trending in difference-in-differences? A synthesis of the recent econometrics literature},
  author={Roth, Jonathan and Sant’Anna, Pedro HC and Bilinski, Alyssa and Poe, John},
  journal={Journal of Econometrics},
  year={2023},
  publisher={Elsevier}
}

@article{roth2023parallel,
  title={When is parallel trends sensitive to functional form?},
  author={Roth, Jonathan and Sant'Anna, Pedro HC},
  journal={Econometrica},
  volume={91},
  number={2},
  pages={737--747},
  year={2023},
  publisher={Wiley Online Library}
}

@article{rudolph2023efficiently,
  title={Efficiently transporting average treatment effects using a sufficient subset of effect modifiers},
  author={Rudolph, Kara E and Williams, Nicholas T and Stuart, Elizabeth A and Diaz, Ivan},
  journal={arXiv preprint arXiv:2304.00117},
  year={2023}
}

@article{sofer2016negative,
  title={On negative outcome control of unobserved confounding as a generalization of difference-in-differences},
  author={Sofer, Tamar and Richardson, David B and Colicino, Elena and Schwartz, Joel and Tchetgen, Eric J Tchetgen},
  journal={Statistical science: a review journal of the Institute of Mathematical Statistics},
  volume={31},
  number={3},
  pages={348},
  year={2016},
  publisher={NIH Public Access}
}

@article{stuart2015assessing,
  title={Assessing the generalizability of randomized trial results to target populations},
  author={Stuart, Elizabeth A and Bradshaw, Catherine P and Leaf, Philip J},
  journal={Prevention Science},
  volume={16},
  pages={475--485},
  year={2015},
  publisher={Springer}
}

@article{westreich2017transportability,
  title={Transportability of trial results using inverse odds of sampling weights},
  author={Westreich, Daniel and Edwards, Jessie K and Lesko, Catherine R and Stuart, Elizabeth and Cole, Stephen R},
  journal={American journal of epidemiology},
  volume={186},
  number={8},
  pages={1010--1014},
  year={2017},
  publisher={Oxford University Press}
}

@article{xu2023difference,
  title={Difference-in-Differences with Interference: A Finite Population Perspective},
  author={Xu, Ruonan},
  journal={arXiv preprint arXiv:2306.12003},
  year={2023}
}

@article{chiu2022selection,
  title={Selection on treatment in the target population of generalizabillity and transportability analyses},
  author={Chiu, Yu-Han and Dahabreh, Issa J},
  journal={arXiv preprint arXiv:2209.08758},
  year={2022}
}

@article{ackerman2021generalizing,
  title={Generalizing randomized trial findings to a target population using complex survey population data},
  author={Ackerman, Benjamin and Lesko, Catherine R and Siddique, Juned and Susukida, Ryoko and Stuart, Elizabeth A},
  journal={Statistics in medicine},
  volume={40},
  number={5},
  pages={1101--1120},
  year={2021},
  publisher={Wiley Online Library}
}

@article{rudolph2022efficiently,
  title={Efficiently transporting causal direct and indirect effects to new populations under intermediate confounding and with multiple mediators},
  author={Rudolph, Kara E and D{\'\i}az, Iv{\'a}n},
  journal={Biostatistics},
  volume={23},
  number={3},
  pages={789--806},
  year={2022},
  publisher={Oxford University Press}
}

@article{robins1994correcting,
  title={Correcting for non-compliance in randomized trials using structural nested mean models},
  author={Robins, James M},
  journal={Communications in Statistics-Theory and methods},
  volume={23},
  number={8},
  pages={2379--2412},
  year={1994},
  publisher={Taylor \& Francis}
}

@article{van2014higher,
  title={Higher order tangent spaces and influence functions},
  author={van der Vaart, Aad},
  journal={Statistical Science},
  pages={679--686},
  year={2014},
  publisher={JSTOR}
}

@article{rotnitzky2021characterization,
  title={Characterization of parameters with a mixed bias property},
  author={Rotnitzky, Andrea and Smucler, Ezequiel and Robins, James M},
  journal={Biometrika},
  volume={108},
  number={1},
  pages={231--238},
  year={2021},
  publisher={Oxford University Press}
}
\begin{appendix}
\section{Proof of identification results}
\subsection{Proof of Theorem \ref{thm:trans_pt}}\label{app:trans_pt}
First we show identification for the PATT:
\begin{align*} 
     \E[Y_1(1)&-Y_1(0)|A=1, S=0] = \E\{ \E[Y_1(1)-Y_1(0)|W, U, A=1, S=0] |A=1,S=0 \} \\
     &=\E\{ \E[Y_1(1)-Y_1(0)|W, U, A=1, S=1] |A=1,S=0 \}  \\
     &=\E\{ \E[Y_1(1)-Y_1(0)|W, A=1, S=1] |A=1,S=0 \}\\
     &=\E( m_1(W) - m_0(W) |A=1,S=0 )
\end{align*}
where the first equality follow by law of iterated expectation (LIE); the second by Assumptions 
\ref{asn:noint}, \ref{asn:tvi}, \ref{asn:exch_u} \& \ref{asn:pos_u}; the third by Assumption \ref{asn:homo}; and the fourth by Assumptions \ref{asn:noint}-\ref{asn:pt}. Next consider identification of the PATU: 
\begin{alignat*} {3}
     \E[Y_1(1)&-Y_1(0)|A=0, S=0] = \E\{ \E[Y_1(1)-Y_1(0)|W, U, A=0, S=0] |A=0,S=0 \} \\
     &=\E\{ \E[Y_1(1)-Y_1(0)|W, U, A=0, S=1] |A=0,S=0 \} \\
    &=\E\{ \E[Y_1(1)-Y_1(0)|W, U, A=1, S=1] |A=0,S=0 \}  \\
     &=\E\{ \E[Y_1(1)-Y_1(0)|W, A=1, S=1] |A=0,S=0 \}\\
     &=\E( m_1(W) - m_0(W) |A=0,S=0 ) 
\end{alignat*}
where the first equality follows by LIE, the second by Assumptions 
\ref{asn:noint}, \ref{asn:tvi}, \ref{asn:exch_u} \& \ref{asn:pos_u}, the third by Assumptions 
    \ref{asn:noint}, \ref{asn:tvi}, \ref{asn:exch_u_tx} \& \ref{asn:pos_u_tx}, the fourth by \textbf{}, and the fifth by Assumptions \ref{asn:noint}-\ref{asn:pt}. Having identified the PATU and the PATT, the PATE is trivially identified:
\begin{align*}
    \E[Y_1(1)-Y_1(0)|S=0] &= \E[ \E(Y_1(1)-Y_1(0)|A, S=0)|S=0] \\
    &= \E[ \E (m_1(W) - m_0(W) |A, S=0)|S=0] \\
    &= \E[ m_1(W) - m_0(W) |S=0]
\end{align*}
\subsection{Identifying the SATE and SATU under additional assumptions}\label{app:sate}
Under Assumptions \ref{asn:noint}-\ref{asn:pt}, \ref{asn:exch_u_tx}, \ref{asn:pos_u_tx}, and \ref{asn:homo}, the SATE is identified as
\begin{align*}
    \E[Y_1(1)&-Y_1(0)|S=1] = \E\{\E[Y_1(1)-Y_1(0)|W, U, S=1]|S=1\}\\
    &= \E\{\E[Y_1(1)-Y_1(0)|W, U, A=1, S=1]|S=1\}  \\
    &= \E\{\E[Y_1(1)-Y_1(0)|W, A=1, S=1]|S=1\} \\
    &= \E\{m_1(W)-m_0(W)|S=1\},
\end{align*}
the first equality by LIE, the second by Assumptions \ref{asn:exch_u_tx} \& \ref{asn:pos_u_tx}, the third by Assumption \ref{asn:homo}, and the fourth by Assumptions \ref{asn:noint}-\ref{asn:pt}. Additionally the SATU is identified as:
\begin{align*}
    \E[Y_1(1)&-Y_1(0)|A=0, S=1] = \E\{\E[Y_1(1)-Y_1(0)|W, U, A=0, S=1]|A=0, S=1\}  \\
    &= \E\{\E[Y_1(1)-Y_1(0)|W, U, A=1, S=1]|A=0, S=1\} \\
    &= \E\{\E[Y_1(1)-Y_1(0)|W, A=1, S=1]|A=0, S=1\} \\
    &= \E\{m_1(W)-m_0(W)|A=0, S=1\},
\end{align*}
the first equality by LIE, the second by Assumptions \ref{asn:exch_u_tx} \& \ref{asn:pos_u_tx}, the third by Assumption \ref{asn:homo}, and the fourth by Assumptions \ref{asn:noint}-\ref{asn:pt}.

\section{Efficient influence function for transported treatment effects}\label{app:eif}
\numberwithin{equation}{section}

\setcounter{equation}{0}

To ease notation, we let $\Delta Y=Y_1-Y_0.$ We let $P_n\{ h(O) \}=n^{-1}\sum_{i=1}^n h(O_i)$ denote the sample mean of a function $h(\cdot)$ of the observed data $O$, $P$ dnote the true distribution of the observed data, and $\widehat P_n$ denote an estimator of the observed data distribution (which may or may not be the empirical distribution). Here we focus on the statistical estimand 
\begin{align*}\psi(a^*, P) &= \E[ m_1(W) - m_0(W)|A=a^*, S=0 ]\} \\
&= \int [m_1(z)-m_0(w)] f(w|A=a^*, S=0)dw
\end{align*}
where we write $\psi(a^*, P)$ as a function of $P$ to emphasize that it depends on the true observed data distribution.

\subsection{Proof of EIF}
\begin{theorem}\label{thm:if} The efficient influence function for $\psi(a^*)$ is given by
\begin{align*}
D(a^*, O_i, P)&=\frac{I(A_i=1, S_i=1) g_{a^*,0}(W_i)}{f(A=a^*, S=0)g_{1,1}(W_i)}
        \{\Delta Y_i - m_1(W_i) \}\\
    &\quad\quad - \frac{I(A_i=0, S_i=1) g_{a^*,0}(W_i)}{f(A=a^*, S=0)g_{0,1}(W_i)}
        \{\Delta Y_i - m_0(W_i) \} \\
    &\quad\quad + \frac{I(A_i=a^*,S_i=0)}{f(A=a^*, S=0)}\{m_1(W_i) - m_0(W_i)\}  - \psi(a^*)   
\end{align*} 
\end{theorem}

\begin{proof} Let $IF(\theta)$ denote the efficient influence function for parameter $\theta=\theta(P)$. Then we have:
   \footnotesize
   \begin{align*}
    IF\{ \psi(a^*, P) \}&=\int \bigg( IF\{m_1(w)\} - IF\{ m_0(w)\}  \bigg)f(w|A=a^*, S=0)dw\\
    &\qquad + \int \{ m_1(W)-m_0(w) \}IF\{ f(w|A=a^*, S=0) \}dw \\
    &=\int 
        \frac{I(A=1,S=1,W=w)}{f(A=1,S=1, W=w)}
        [\Delta Y - m_1(w)]f(w|A=a^*, S=0)dw \\
        &\qquad- \int 
        \frac{I(A=0,S=1,W=w)}{f(A=0,S=1, W=w)}
        [\Delta Y - m_0(w)]f(w|A=a^*, S=0)dw \\
    &\qquad+ \int \{ m_1(w)-m_0(w) \}\frac{I(A=a^*,S=0)}{f(A=a^*, S=0)}\{ I(W=w) -  f(w|A=a^*, S=0) \}dw \\
    &=\frac{I(A=1, S=1)g_{a^*, 0}(W)}{f(A=a^*, S=0)g_{1,1}(W)}[\Delta Y - m_1(W)] - \frac{I(A=0, S=1)g_{a^*, 0}(W)}{f(A=a^*, S=0)g_{0,1}(W)}[\Delta Y - m_0(W)] \\
    &\qquad + \frac{I(A=a^*,S=0)}{f(A=a^*,S=0)}[m_1(W)-m_0(W)] - \int \{ m_1(w)-m_0(w) \}f(w|A=a^*, S=0)dw
\end{align*}
\normalsize
where the first equality uses the fact that the efficient influence function is a derivative and applies the product rule for derivatives, the second substitutes known influence functions for conditional expectations and conditional densities/probability mass functions, and the third rearranges.
\end{proof}
\subsection{One-step estimator}
Because the efficient influence function for $\psi(a^*)$ is given by $D(a^*, O_i, P)$, a one-step estimator \cite{hines2022demystifying} is given by 
\begin{align*}
    \widehat\psi_{dr}(a^*)&=\psi(a^*, \widehat P_n) - P_n \{ D(a^*, O_i, \widehat P_n) \}\\
    &= P_n \bigg\{ \frac{I(A=1, S=1) \widehat g_{a^*,0}(W)}{ P_n( I[A=a^*, S=0] )\widehat g_{1,1}(W)}
        \{\Delta Y - \widehat m_1(W) \} \\
    &\qquad - \frac{I(A=0, S=1) \widehat g_{a^*,0}(W)}{P_n( I[A=a^*, S=0] ) \widehat g_{0,1}(W)}
        \{\Delta Y - \widehat m_0(W) \} \\
    &\qquad + \frac{I(A=a^*,S=0)}{ P_n( I[A=a^*, S=0] )}\{\widehat m_1(W) - \widehat m_0(W)  \} \bigg\}
\end{align*}
where $\widehat P_n=\{\widehat g_{a, s}, \widehat m_a\}$ with $\{a, s\}\in \{0,1\}^2.$ If $\widehat P_n$ consists of correctly-specified parametric models, it follows by the central limit theorem that
\[
\sqrt{n}[\widehat\psi_{dr}(a^*) - \psi(a^*)] \xrightarrow{d} N(0, \E[ D(a^*, O, \widehat P_n)^2 ] )
\]
where $\xrightarrow{d}$ denotes convergence in distribution. Thus, an estimator of the asymptotic variance that is consistent when $\widehat P_n$ consists of correctly-specified parametric models is given by 
\[
\widehat{var}[\widehat\psi_{dr}(a^*)]=P_n \{  D(a^*, O, \widehat P_n)^2 \}
\]

\section{Double robust property of $\widehat\psi_{dr}$}\label{app:dr}
In this section, for ease of notation we let $m_a\equiv m_a(W)$ and $g_{a,s}\equiv g_{a,s}(W)$. First, suppose $\widehat m_a \xrightarrow{p} m_a^*$ and $\widehat g_{a, s} \xrightarrow{p} g_{a, s}^*$, not necessarily assuming $g_{a, s} = g_{a, s}^*$ or $m_a = m_a^*$. Using Slutzky's theorem and continuous mapping theorem we have
\small
\begin{align*}
    \widehat \psi_{dr}(a^*) &\xrightarrow{p} 
    \psi^*(a^*) \equiv \E \bigg \{ 
         \frac{I(A=1, S=1) g_{a^*,0}^*}{f(A=a^*, S=0) g_{1,1}^*}
        \{\Delta Y - m_1^* \} - \frac{I(A=0, S=1) g_{a^*,0}^*}{f(A=a^*, S=0)g_{0,1}^*}
        \{\Delta Y - m_0^*\} \\
    &\quad\quad\quad + \frac{I(A=a^*,S=0)}{f(A=a^*, S=0)}\{ m_1^* - m_0^* \} \
    \bigg \} 
\end{align*}
\normalsize
First consider the case where outcome models are correctly specified, so that $m_a(w)= m_a^*(w)$ for all $w\in \mathcal{W}.$ By the linearity property of expectations, we have
\begin{align}
    \psi^*(a^*) &= \E \bigg \{ 
         \frac{I(A=1, S=1) g_{a^*,0}^*}{f(A=a^*, S=0) g_{1,1}^*}(\Delta Y - m_1^*) \bigg \} \label{eq:dr1}\\
         &\quad\quad - \E\bigg\{\frac{I(A=0, S=1) g_{a^*,0}^*}{f(A=a^*, S=0)g_{0,1}^*}(\Delta Y - m_0^*) \bigg\} \label{eq:dr2}\\
         &\quad\quad + \E\bigg\{\frac{I(A=a^*,S=0)}{f(A=a^*,S=0)}(m_1^*-m_0^*)\bigg\} \label{eq:dr3}
\end{align}
Under $m_a= m_a^*,$ we have
\begin{align*}
    (\ref{eq:dr1})&=\E\bigg \{
        \frac{I(A=1,S=1)g_{a^*,0}^*}{f(A=a^*,S=0)g_{1,1}^*}
        (\E[\Delta Y |A=1,S=1,W] - m_1 ) 
    \bigg\} \\
    &=0
\end{align*}
Likewise, $(\ref{eq:dr2})=0.$ Therefore, $\psi^*(a^*)=\E\bigg\{\frac{I(A=a^*,S=0)}{f(A=a^*,S=0)}(m_1-m_0)\bigg\}=\E[m_1-m_0|A=a^*,S=0]=\psi(a^*).$ Moreover, because (\ref{eq:dr3}) is the probability limit of the g-computation estimator, we have that $\widehat\psi_{gcomp}(a^*)\xrightarrow{p}\psi(a^*)$ whenever $\widehat m_a \xrightarrow{p} m_a.$ \\

Next, consider the case where treatment and selection models are correctly specified, so that $g_{a,s}(w)=g_{a,s}^*(w)$ for all $w \in \mathcal{W}$. Rearranging we have
\begin{align}
    \psi^*(a^*) &= \E \bigg \{ 
        \bigg (
         \frac{I(A=1, S=1) g_{a^*,0}^*}{f(A=a^*, S=0) g_{1,1}^*} - \frac{I(A=0, S=1) g_{a^*,0}^*}{f(A=a^*, S=0)g_{0,1}^*}
        \bigg )  \Delta Y
        \bigg \} \label{eq:dr4}\\
    & \quad + \E \bigg\{
        \bigg (
            \frac{I(A=0,S=1)g_{a^*,0}^*}{f(A=a^*,S=0)g_{0,1}^*}
            -
            \frac{I(A=1, S=0)}{f(A=a^*, S=0)}
        \bigg ) m_0^*
    \bigg\} \label{eq:dr5}\\
    & \quad - \E \bigg\{
        \bigg (
            \frac{I(A=1,S=1)g_{a^*,0}^*}{f(A=a^*,S=0)g_{1,1}^*}
            -
            \frac{I(A=a^*, S=0)}{f(A=a^*, S=0)}
        \bigg ) m_1^*
    \bigg\} \label{eq:dr6}
\end{align}
Since $g_{a, s}^* = g_{a, s}$, we have
\footnotesize
\begin{align*}
    (\ref{eq:dr4}) &= \E \bigg \{ 
        \bigg (
         \frac{I(A=1, S=1) f(A=a^*,S=0|W)}{f(A=a^*, S=0) f(A=1,S=1|W)} - \frac{I(A=0, S=1) f(A=a^*,S=0|W)}{f(A=a^*, S=0)f(A=0,S=1|W)}
        \bigg )  \Delta Y
        \bigg \}\\
        &= \E \bigg \{ 
         \frac{I(A=1, S=1) f(A=a^*,S=0|W)}{f(A=a^*, S=0) f(A=1,S=1|W)}m_1(W) - \frac{I(A=0, S=1) f(A=a^*,S=0|W)}{f(A=a^*, S=0)f(A=0,S=1|W)}m_0(W)
        \bigg \} \\
        &= \E \bigg \{ 
         \frac{\E[I(A=1, S=1)|W] f(A=a^*,S=0|W)}{f(A=a^*, S=0) f(A=1,S=1|W)}m_1(W) \\
        &\qquad \qquad- \frac{\E[I(A=0, S=1)|W] f(A=a^*,S=0|W)}{f(A=a^*, S=0)f(A=0,S=1|W)}m_0(W)
        \bigg \} \\
        &= \E \bigg \{ 
         \frac{f(A=a^*,S=0|W)}{f(A=a^*, S=0)}[m_1(W) - m_0(W)]
        \bigg \} \\
        &= \E \bigg \{ 
         \frac{I(A=a^*, S=0)}{f(A=a^*, S=0)}[m_1(W) - m_0(W)]
        \bigg \} \\
        &=\psi(a^*)  \\
    (\ref{eq:dr5}) &=\E\bigg \{
    \bigg (
        \frac{\E[I(A=0, S=1)|W] f(A=a^*,S=0|W)}{f(A=a^*,S=0)f(A=0, S=1|W)} - \frac{I(A=a^*,S=0)}{f(A=a^*,S=0)} 
    \bigg ) m_0^* 
        \bigg \} \\
    &=\E\bigg \{
    \bigg (
        \frac{f(A=a^*,S=0|W)}{f(A=a^*,S=0)} - \frac{I(A=a^*,S=0)}{f(A=a^*,S=0)} 
    \bigg ) m_0^* 
        \bigg \} \\
    &=0 
\end{align*}
\normalsize
Likewise, (\ref{eq:dr6})=$\E\bigg\{ \bigg ( \frac{f(A=a^*,S=0|W)}{f(A=a^*,S=0)} - \frac{I(A=a^*,S=0)}{f(A=a^*,S=0)} \bigg ) m_1^* \bigg\}=0.$  Thus if $\widehat g_{a, s}(w) \xrightarrow{p} g_{a, s}(w)$, $\widehat \psi_{dr}(a^*) \xrightarrow{p} \psi(a^*).$ Moreover, because (\ref{eq:dr4}) is the probability limit of the IOW estimator, we have that $\widehat\psi_{IOW}(a^*)\xrightarrow{p}\psi(a^*)$ whenever $\widehat g_{a,s} \xrightarrow{p} g_{a,s}$.

\end{appendix}
\end{document}